\documentclass[11pt]{article}
\usepackage[letterpaper, margin=1in]{geometry}
\usepackage{enumitem}
\usepackage{mathtools}
\usepackage{amsthm}
\usepackage{amssymb}
\usepackage{xspace}
\usepackage{bm}

\usepackage[breakable]{tcolorbox}
\usepackage{longtable}
\usepackage{booktabs}
\usepackage{graphicx}
\usepackage{relsize}
\usepackage[left, pagewise]{lineno}
\usepackage{tikz}
  \usetikzlibrary{positioning,fit,calc}
  \tikzstyle{block} = [draw=black, ultra thin, text width=2cm, minimum height=1.3cm, font = {\footnotesize\itshape},align=center]
  \tikzstyle{arrow} = [thick,->,>=stealth]

  \renewenvironment{quote}
  {\list{}{\leftmargin=0.5cm \rightmargin=0.5cm} 
   \item\relax}
  {\endlist}

  \usepackage[page,toc,titletoc,title]{appendix}
\usepackage[textsize=footnotesize]{todonotes}
\usetikzlibrary{shapes,arrows,positioning,patterns}
\usepackage{comment}
\usepackage{hyperref}
\usepackage{cleveref}
\hypersetup{
    colorlinks=true,
    citecolor=blue,
    linkcolor=blue,
    filecolor=blue,
    urlcolor=blue,
}
\usepackage{thm-restate}
\usepackage{placeins}
\usepackage[font=small,labelfont=bf,margin=0.06in]{caption}

\usepackage[style=alphabetic,maxnames=999]{biblatex}
\AtEveryBibitem{
  \clearfield{editor}
  \clearlist{editor}
  \clearname{editor}
  \clearfield{location}
  \clearlist{location}
  \clearname{location}
  \clearfield{isbn}
  \clearlist{isbn}
  \clearname{isbn}
  \clearfield{url}
  \clearlist{url}
  \clearname{url}
}

\RequirePackage[osf,sc]{mathpazo}
\usepackage{framed}
\usepackage{mdframed}
\usepackage{float}
\usepackage{xcolor}

\newcommand{\ind}[2][]{%
  \mathrel{
    \mathop{
      \vcenter{
        \hbox{\oalign{\noalign{\kern-.3ex}\hfil$\vert$\hfil\cr
              \noalign{\kern-.7ex}
              $\smile$\cr\noalign{\kern-.3ex}}}
      }
    }^{#2}\displaylimits_{#1}
  }
}

\newcommand{\eqdef}{\coloneqq}

\newlength{\leftbarwidth}
\newlength{\leftbarsep}
\renewenvironment{leftbar}[1][blue]
{%
\def\FrameCommand
{%
{\hspace{-8pt} \color{#1} \vrule width 3pt}%
\hspace{0pt}
\fboxsep=\FrameSep\colorbox{#1!5!}%
}%
\MakeFramed{\hsize\hsize\advance\hsize-\width\FrameRestore}%
}
{\endMakeFramed}
\makeatletter

  \if@todonotes@disabled 
  \excludecomment{szin}
  \else 
  
  \fi
  \if@todonotes@disabled 
  \excludecomment{janin}
  \else 
  
  \fi

\makeatother

\newcommand\myitem[1]{\hyperref[item:#1]{\emph{\ref*{item:#1}.}{}}\xspace}

\newtheorem{theorem}{Theorem}[section]
\newtheorem*{theorem*}{Theorem}

\newtheorem{corollary}[theorem]{Corollary}
\newtheorem*{corollary*}{Corollary}
\newtheorem{lemma}[theorem]{Lemma}
\newtheorem*{lemma*}{Lemma}

\newtheorem{observation}[theorem]{Observation}

\newtheorem*{proposition*}{Proposition}

\crefname{claim}{claim}{Claims} 
\Crefname{claim}{Claim}{Claims} 

\newenvironment{claimproof}[1][\proofname]{%
  \begin{proof}[#1]%
}{%
  \end{proof}%
}

\theoremstyle{remark}

\newcommand{\from}{\colon}

\newcommand{\set}[1]{\{#1\}}

\newcommand{\setof}[2]{\set{#1 \,|\, #2 }}

\def\phi{\varphi}
\def\cal{\mathcal}
\def\N{\mathbb N}

\def\epsilon{\varepsilon}

\renewcommand{\subset}{\subseteq}

\renewcommand{\le}{\leqslant}
\renewcommand{\ge}{\geqslant}

\newcommand{\dist}{\mathrm{dist}}
\newcommand{\distFO}{\textrm{\upshape dist-FO}\xspace}

\newcommand{\tup}{\bar}

\newcommand{\ERCagreement}{\xspace ST received funding from the European Research Council (ERC) (grant agreement  №101126229 -- {\sc buka}).
\begin{tikzpicture}[remember picture, overlay]
  \coordinate (refpoint) at (12,0); 

  \node[anchor=south, yshift=0cm] at (refpoint)
  {\includegraphics[width=40px]{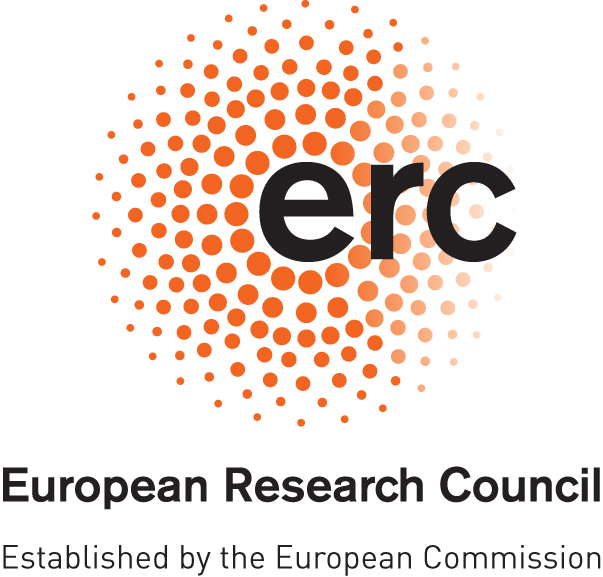}};

  \node[anchor=south, yshift=-2cm] at (refpoint)
  {\includegraphics[width=60px]{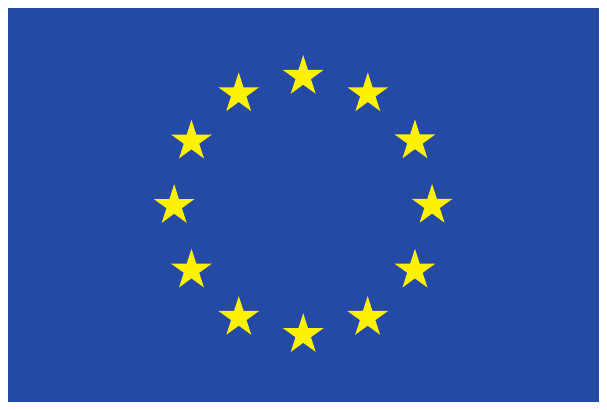}};
\end{tikzpicture}
}

    \title{A Rank-Preserving Locality Theorem}
\author{Jan Dreier \and Szymon Toru{\'n}czyk}
\date{\today}

\begin{document}\maketitle
\begin{abstract}
    We prove a rank-preserving locality theorem for a syntactic variant of first-order logic, in the spirit of Gaifman's locality theorem and the rank-preserving locality theorem of Grohe, Kreutzer, and Siebertz. 
    Our result allows for a weak form of scatter sentences, which can be evaluated more efficiently than usual scatter sentences considered in prior work. 
    This is crucial in our application to graphs of bounded merge-width.
\end{abstract}

\paragraph{Acknowledgements.}
We are grateful to Nikolas M\"ahlmann for many discussions, and spotting an error in \cite{mw-arxiv-v1}. \ERCagreement

\section{Introduction}\label{sec:intro}
The goal of this note is to prove the following locality result for a variant of first-order logic. All relevant notions are defined in detail in \Cref{sec:preliminaries}, and discussed briefly below.

\begin{restatable}[Locality theorem for \distFO]{theorem}{localglobalgame}\label{thm:localglobal}
    Fix $k,q \in \N$. Every \distFO formula $\phi(\bar x)$ of distance rank $(k,q)$
     is equivalent to a boolean combination of
    local \distFO formulas and \distFO scatter sentences, all of distance rank $(k,q)$.    
    This boolean combination can be effectively computed, given $k,q$, and $\phi$.
\end{restatable}

The logic \distFO is equivalent in expressive power to first-order logic, but its syntax allows for a more refined measure of quantifier rank, called \emph{distance rank}. \emph{Local} \distFO formulas are a fragment which only allow for \emph{local quantification}, which ranges over a neighborhood of the free variables of fixed radius.

This statement is in the spirit of Gaifman's locality theorem \cite{gaifman},
and is 
almost the same as a statement in our earlier work \cite[Theorem 4.1]{mw-stoc,mw-arxiv-v1}.
In this note we fix an error in the proof from \cite{mw-arxiv-v1}, which was spotted by Nikolas M\"ahlmann. Apart from that, we also simplify the proof significantly.

Our result is inspired by a locality theorem of Grohe, Kreutzer, and Siebertz  \cite{gks}.
Similarly as in their result, the challenge is to 
rewrite a formula into a boolean combination of scatter sentences and local formulas in a way which does not increase some notion of rank. 
This is achieved by changing the syntax, by introducing \emph{distance atoms} and \emph{local quantification}.
The original motivation for those results was to prove that 
the problem of evaluating first-order formulas is fixed-parameter tractable in various graph classes (nowhere dense classes in \cite{gks} and classes of bounded merge-width in \cite{mw-stoc}).

\paragraph{Comparison with the work of Grohe, Kreutzer, and Siebertz.}
Comparing \Cref{thm:localglobal} to the locality theorem  \cite[Theorem 7.5]{gks} of Grohe, Kreutzer, and Siebertz, there are three main differences.
The first one is that we consider a weak notion of \emph{scatter sentences}. In Gaifman's result, as well as the result of Grohe, Kreutzer, and Siebertz, scatter sentences assert the existence of $t$ elements 
of some set $X$ which are mutually at distance greater than $r$ in the Gaifman graph, where $r$ and $t$ are fixed constants, and $X$ is specified by a formula $\alpha(x)$. Determining this is algorithmically hard in some contexts -- from the perspective of parametrized complexity, this problem is at least as hard as the distance-$r$ independent set problem, which thus needs to be solved separately \cite[Theorem 5.1]{gks}.
 In contrast, the scatter sentences in \Cref{thm:localglobal} 
merely ask whether a predetermined inclusion-wise maximal $r$-scattered subset of $X$ has size at least~$t$. Such a set can be easily computed by a greedy algorithm. This is crucial in our algorithmic application of \Cref{thm:localglobal} to graphs of bounded merge-width, as it allows to circumvent the independent set problem.

Next, the result of \cite{gks} evaluates the resulting formulas in a structure expanded by suitable unary predicates, whereas our resulting formulas are equivalent to the input formulas on any given (unmodified) structure. This is also crucial in our application to graphs of bounded merge-width \cite{mw-stoc}.

Finally, the result of \cite{gks} applies only to formulas $\phi(x)$ with one free variable, whereas in our statement, $\phi(\bar x)$ is a formula with any set $\bar x$ of free variables.



\paragraph{Comparison with the work of Grohe and Schweikardt.}
Independently of this work, Grohe and Schweikardt have very recently published a note \cite{grohe2026rankpreservinggaifmannormalform}.
The results of \cite{grohe2026rankpreservinggaifmannormalform} use a very similar syntactic fragment of first-order logic to ours, and measure the distance rank of formulas in a very similar way. The overall goal is also very similar -- to provide a locality theorem for first-order logic, in which formulas are expressed as boolean combinations of local formulas and scatter sentences,
and where the rewriting preserves some notion of \emph{rank} (and also 
to fix an error in a previous paper \cite{gs-pods}).

Overall, the results of the two notes are incomparable, however, two results --
\Cref{cor:normalform} below and the core result \cite[Theorem 6.3]{grohe2026rankpreservinggaifmannormalform} -- are in spirit the same statements, although minor details differ. 
Some differences in the two notes are as follows:
\begin{enumerate}
    \item Our result considers more flexible scatter sentences, which express that a single, arbitrarily chosen, inclusion-wise maximal $r$-scattered set has a given size, while \cite{grohe2026rankpreservinggaifmannormalform} considers scatter sentences which express the existence of an $r$-scattered set of a given size. The former is more general and allows for more efficient algorithmic applications, circumventing the independent set problem. In particular, the more flexible form of scatter sentences allows us to apply the locality theorem in the setting of model-checking of first-order formulas on graphs of bounded merge-width \cite{mw-arxiv-v1,mw-stoc}.
    \item The horizon function in our note is of the order $2^{\Theta(q^2)}$, which is larger than the horizon function in \cite{grohe2026rankpreservinggaifmannormalform}, which is of the order $2^{\Theta(q)}$. 
    This seems to be related to the more 
flexible form of scatter sentences in this paper.

    \item The note \cite{grohe2026rankpreservinggaifmannormalform} derives further locality results (Theorems 7.1 and  8.2), which have no analogue in our note.
    \item The proof in \cite{grohe2026rankpreservinggaifmannormalform} is based on Ehrenfeucht-Fra\"iss\'e games and types, similarly to our original approach in \cite{mw-arxiv-v1}. The present proof is based on a syntactic rewriting of formulas, which we find more direct, transparent, and concise. The two approaches are ultimately equivalent, but 
    we believe that the game-based approach is more prone to errors (as have been spotted in \cite{mw-arxiv-v1,gs-pods}), since tracing invariants for games turns out to be often subtle.

\item Finally, the syntactic rewriting approach is algorithmic, and allows for an effective procedure to compute the boolean combination of local formulas and scatter sentences. 
    \cite{grohe2026rankpreservinggaifmannormalform} are also able to derive -- a posteriori -- the existence of an algorithm computing the equivalent boolean combination, by appealing to G\"odel's completeness theorem for first-order logic. This approach does not provide any upper bound on the duration of the algorithm. Our syntactic approach provides an explicit algorithm whose running time can in principle be analyzed and upper bounded by an explicit computable function (which is enormous but elementary for fixed $q$ and $k$), although we refrain from doing so in this note.
\end{enumerate}

\paragraph{Organization of the paper.}
In \Cref{sec:preliminaries} we introduce the notions needed to state the locality theorem, \Cref{thm:localglobal}.
We derive \Cref{cor:normalform} from it, which gives a rank-preserving normal form for \distFO formulas. In \Cref{sec:proof} we prove the locality theorem.

\section{Preliminaries}\label{sec:preliminaries}
 $\N$ denotes the set of non-negative integers. For $n\in\N$, denote $[n]\eqdef\set{1,\ldots,n}$.

We consider finite relational signatures only.
A structure $A$ over the signature $\sigma$,
 or a \emph{$\sigma$-structure}, consists of a (possibly infinite) domain $V(A)$ and the interpretation
  $R^A\subset V(A)^k$ of each relation symbol $R\in \sigma$ of arity $k$.
The \emph{Gaifman graph} of a structure $A$ is the graph with vertices $V(A)$ and edges connecting pairs of distinct vertices which occur together in some tuple of some relation of $A$. By $\dist^A(\cdot,\cdot)$ we denote the shortest path metric  in the Gaifman graph of $A$, and for $r\in\N$ and $u\in V(A)$ denote $N^A_r(u)\coloneqq\setof{v\in V(A)}{\dist^A(u,v)\le r}$.




\subsection{Distance logic and scatter sentences}\label{sec:typedefs}

\paragraph{First-order distance logic.}
We define a logic called \distFO that extends FO by the following \emph{distance atoms}.
\begin{quote}
     For each radius \(r \in \N\), \distFO introduces
        a \emph{binary distance atom} \(\dist(x,y) \le r\) expressing that the distance between \(x\) and \(y\) in the Gaifman graph is at most \(r\).
        We call \(r\) the \emph{radius} of the distance atom.

     For each radius \(r \in \N\) and unary relation symbol \(Y\), \distFO introduces
        a \emph{unary distance atom} \(\dist(x,Y) < r\) expressing that the distance from \(x\) to \(Y\) in the Gaifman graph is smaller than \(r\);
        equivalently, \(\exists y~Y(y) \land \dist(x,y) < r\).
        We call \(r\) the \emph{radius} of the distance atom.
\end{quote}
Formulas of \distFO are built inductively,
starting with usual atomic formulas of first-order logic (relation symbols or equality applied to variables), distance atoms, boolean connectives, and existential quantification (universal quantification is then expressible using negation).

For a \distFO formula $\phi(x_1,\ldots,x_k)$, structure $A$, 
and  tuple $(a_1,\ldots,a_k)\in V(A)^k$
we write $A\models\phi(a_1,\ldots,a_k)$
to denote
that $\phi(a_1,\ldots,a_k)$ holds in $A$, which is defined by induction on the structure of $\phi$, in the expected way.

For a set $\bar x$ of variables, 
we write $\phi(\bar x)$ to indicate that $\bar x$ contains the free variables of the formula $\phi$.
We write \(\dist(\bar x,y) \le r\) as shorthand for \(\bigvee_{x \in \bar x}\dist(x,y) \le r\), and
\(\dist(\bar x,\bar y) \le r\) as shorthand for
\(\bigvee_{x\in \bar x, y \in \bar y}\dist(x,y) \le r\).

Note that binary distance atoms can be expressed
using usual first-order formulas (assuming~$\sigma$ contains finitely many relation symbols of arity ${>}1$).
However, we treat them as primitive atomic formulas, as
we will need to have a fine control over the quantifier rank of the constructed formulas. For this reason, the possibility to measure distances using quantifier-free formulas is essential, similarly as in \cite{gks}.

\paragraph{Horizon functions.}
Fix two \emph{horizon functions} $\rho^-,\rho^+\from \N^2\to \N_{\ge 1}$ satisfying the following properties for all $k,q\in\N$:
\begin{equation}\label{eq:main-property-rho}
\begin{aligned}
    \rho^-(k,q)&\ge \rho^+(k+1,q-1)+\rho^-(k+1,q-1)
    &&\quad\text{for $q\ge 1$},\\
    \rho^+(k,q)&\ge 9^k \cdot \rho^-(k,q).
\end{aligned}
\end{equation}
For instance, it is easy to check that the following
functions satisfy \eqref{eq:main-property-rho}:
\begin{align*}
   \rho^-(k,q) \eqdef 9^{(k+q+1)q} \quad\quad\quad\text{and}\quad\quad\quad
   \rho^+(k,q) \eqdef 9^{(k+q)(q+1)}.
\end{align*}
The functions \(\rho^-\) and \(\rho^+\) will be used to define the allowed radius of distance atoms in \distFO formulas and scatter sentences of a given \emph{distance rank} (defined below), as well as the radius at which \emph{local quantification} is allowed.

\paragraph{Distance rank.}
Let \(k,q \in \N\).
We define \emph{\distFO formulas of distance rank} \((k,q)\) by induction on~$q$,
as formulas with at most $k$ free variables and
quantifier rank at most $q$ which are boolean combinations of:
\begin{itemize}
    \item atoms of first-order logic (relation symbols in $\sigma\cup\set{=}$ applied to variables),
    \item unary and binary distance atoms with radius  \(r\le \rho^-(k,q)\),
    \item if $q\ge 1$, formulas
        \[
            \exists y~\phi(\bar x y)
            \quad\quad\text{or}\quad\quad
            \exists y~\bigl(\dist(\bar x,y) \le \rho^+(k+1,q-1)\bigr) \land \phi(\bar x y)
        \]
        where \(\phi(\bar x y)\) has distance rank \((k+1,q-1)\).
\end{itemize}
A \distFO formula is \emph{local}
if it never uses the \emph{unrestricted quantification} $\exists y~\phi$
(only \emph{local quantification}, that is, of the latter form above, is allowed).

\medskip

Property \eqref{eq:main-property-rho} implies \(\rho^-(k+1,q-1)\le \rho^+(k+1,q-1) < \rho^-(k,q)\) for all \(k,q \in \N\) with $q\ge 1$. Hence, the distance atoms have progressively shorter radii as one traverses the quantifiers of a \distFO formula inwards. It moreover implies:

\begin{observation}\label{obs:preservedistrank}
    Every \distFO formula of distance rank \((k+1,q-1)\) with at most \(k\) free variables also has distance rank \((k,q)\).
\end{observation}

\paragraph{Semantic locality.}
We first prove a property of local \distFO formulas, which states that their truth value is determined by a bounded-radius neighborhood of the free variables. We say that a formula $\phi(\bar x)$ is \emph{semantically $r$-local}
if for every structure $A$ and tuple $\bar a\in A^{\bar x}$,
we have that
$$A\models \phi(\bar a) \quad\Longleftrightarrow\quad A[N_r(\bar a)]\models \phi(\bar a).$$

\begin{lemma}\label{lem:ltp-local-subgraph}
    Let \(k,q \in \N\) with $k\ge 1$, and let
    \(\phi(\bar x)\) be a local \distFO formula of distance rank \((k,q)\).
    Then $\phi$ is semantically \(\rho^-(k,q)\)-local.
\end{lemma}
\begin{proof}
    We use the following consequence of \eqref{eq:main-property-rho}. For every
    \(j\in\{0,\ldots,q\}\),
    \[
        \sum_{h=1}^{j}\rho^+(k+h,q-h)+\rho^-(k+j,q-j)
        \le \rho^-(k,q),
    \]
    where the case \(j=0\) is trivial, and the case \(j>0\) follows by applying
    \[
        \rho^+(k+h,q-h)
        \le
        \rho^-(k+h-1,q-h+1)-\rho^-(k+h,q-h)
    \]
    for \(h=1,\ldots,j\) and telescoping.

    Now fix a structure \(A\) and a tuple \(\bar a\in A^{\bar x}\). If, during
    the evaluation of \(\phi(\bar a)\), a variable is introduced after \(j\)
    local quantifiers, then it is at distance at most
    \(\sum_{h=1}^{j}\rho^+(k+h,q-h)\) from \(\bar a\). Any distance atom in the
    scope of these \(j\) quantifiers has radius at most \(\rho^-(k+j,q-j)\).
    By the displayed bound, all vertices that can affect the truth of such an atom
    lie in \(N_{\rho^-(k,q)}^A(\bar a)\). Hence removing all vertices outside this
    neighborhood does not change the truth value of \(\phi(\bar a)\).
\end{proof}

\paragraph{Scatter sentences.}
A set \(S\) of vertices of a structure $A$ is \emph{\(r\)-scattered} if all vertices in \(S\)
have pairwise distance larger than \(r\) in the Gaifman graph of $A$. (Thus, a \(1\)-scattered set is an independent set.)

For every structure \(A\), radius \(r\), formula \(\beta(x)\)
we fix arbitrarily and once and for all
fix a value $s_{A,r,\beta}\in\N\cup\set{\infty}$ such that there is some inclusion-wise maximal \(r\)-scattered subset of \(\{ a \in V(A) \mid A \models \beta(a)\}\) of size $s_{A,r,\beta}$; the value $\infty$ indicates that there exist such sets of arbitrarily large finite size. 

    Two possible choices of $s_{A,r,\beta}$ are described below. The particular choice is not important for the statement of \Cref{thm:localglobal}, but it is relevant in applications of the result, as it might affect the efficiency of evaluation of scatter sentences.

Our locality theorem relies on the fixed values $s_{A,r,\beta}$.
To incorporate it into logical statements, we define logical \emph{scatter sentences}
\(\textbf{scatter}(r,\beta(x),t)\) with \(A \models \textbf{scatter}(r,\beta(x),t)\) if and only if $s_{A,r,\beta}\ge t$.
We define \emph{scatter sentences of distance rank \((k,q)\)}, for \(q\ge 1\),
as all scatter sentences \(\mathbf{scatter}(r,\beta(x),t)\) such that
\(t\le k+q\) and
for some \(i\in\{1,\ldots,q\}\),
\(\beta(x)\) is a local \(\distFO\) formula of distance rank \((k+i,q-i)\), and 
\begin{align}\label{eq:scatter-radius}
4\rho^-(k+i,q-i)\le r\le 9^{k+i}\rho^-(k+i,q-i).
\end{align}
Note that
$9^{k+i}\rho^-(k+i,q-i)\le \rho^-(k,q)$, which follows from \eqref{eq:main-property-rho}: first
\[
    9^{k+i}\rho^-(k+i,q-i)
    \le \rho^+(k+i,q-i)
    \le \rho^-(k+i-1,q-i+1),
\]
and then we iterate \(\rho^-(a,b)\le \rho^+(a,b)\le \rho^-(a-1,b+1)\).
Moreover, by \Cref{lem:ltp-local-subgraph}, the formula \(\beta(x)\) as above
semantically \(\rho^-(k+i,q-i)\)-local, and hence semantically \(r/4\)-local.

We have the following.

\begin{observation}\label{obs:preservescatterrank}
    Every scatter sentence of distance rank \((k+1,q-1)\) is also a scatter sentence of distance rank \((k,q)\).
\end{observation}

We now describe two possible choices of $s_{A,r,\beta}$, each having its own benefits. The first one is more algorithmic, while the second one is more logical.

\begin{description}
    \item[Greedy choice]
    For finite structures $A$, an inclusion-wise maximal subset of \(X\coloneqq\{ a \in V(A) \mid A \models \beta(a)\}\) can be constructed by a simple greedy process as follows:
Assuming \(v_1,\dots,v_i\) have already been chosen, the process selects \(v_{i+1}\) as the next element of $X$ (with respect to some arbitrary, fixed order on \(V(A)\))
 and is of distance larger than \(r\) in the Gaifman graph of $A$ from \(v_1,\dots,v_i\).
If no such \(v_{i+1}\) can be found, the process terminates and we set $s_{A,r,\beta} \coloneqq i$.
The benefit of this choice of $s_{A,r,\beta}$ is that it can be efficiently computed, given $A$, an order on $V(A)$, and $X$.
    \item[Maximum size] 
     Another way to define \(s_{A,r,\beta}\) is to set it to the maximum size of an \(r\)-scattered subset of \(\{ a \in V(A) \mid A \models \beta(a)\}\) and to $\infty$ if arbitrarily large such sets exist. The benefit of this definition is that it does not rely on the choice of an arbitrary order on $V(A)$, and for each $t\in\N$, the  property $s_{A,r,\beta}\ge t$ can be expressed in first-order logic (for fixed $r$ and $t$), by the following sentence: 
     \begin{align}\label{eq:scatt-sent}
        \exists x_1,\ldots,x_t\,
\bigwedge_{i\neq j}\left(\dist(x_i,x_j)>r\right)
\ \land\
\bigwedge_{i=1}^t \beta(x_i).
     \end{align}

\end{description}

\paragraph{Locality theorem.}
We have now stated all the notions occurring in \Cref{thm:localglobal}, repeated below for convenience.
\localglobalgame*

As a corollary of \Cref{thm:localglobal}, we obtain a normal form for \distFO formulas, which is an analogue of Gaifman's normal form for first-order logic \cite{gaifman}. 
The proof is a straightforward application of \Cref{thm:localglobal} and the definition of scatter sentences.

\begin{corollary}[Normal form for \distFO formulas]\label{cor:normalform}
Fix \(k,q\in\N\) with \(q\ge 1\). Every \(\distFO\) formula
\(\phi(\bar x)\) of distance rank \((k,q)\) is equivalent to a boolean
combination of local \(\distFO\) formulas of distance rank \((k,q)\)
and sentences of the form
\[
\exists x_1,\ldots,x_t\,
\bigwedge_{i\neq j}\left(\dist(x_i,x_j)>r\right)
\ \land\
\bigwedge_{i=1}^t \beta(x_i).
\]
where \(t\le k+q\), the formula
\(\beta(x)\) is a local \(\distFO\) formula of distance rank
\((k+1,q-1)\), the radius \(r\) satisfies \(r\le \rho^-(k,q)\),
and \(\beta(x)\) is semantically \(r/4\)-local. This boolean combination can be effectively computed, given
\(k,q\), and~\(\phi\).
\end{corollary}
\begin{proof}
    In the definition of scatter sentences, choose \(s_{A,r,\beta}\) to be the
    maximum size of an \(r\)-scattered subset of
    \(\{a\in V(A)\mid A\models \beta(a)\}\), as in the ``Maximum size'' choice
    above. With this choice, for every \(r,t\in\N\) and every formula
    \(\beta(x)\), the scatter sentence \(\textbf{scatter}(r,\beta(x),t)\) is
    equivalent to the sentence \eqref{eq:scatt-sent}.

    Apply \Cref{thm:localglobal} to \(\phi(\bar x)\). We obtain an effectively
    computed boolean combination of local \(\distFO\) formulas of distance rank
    \((k,q)\) and scatter sentences of distance rank \((k,q)\). Replace each
    scatter sentence \(\textbf{scatter}(r,\beta(x),t)\) in this boolean
    combination by the equivalent  sentence~\eqref{eq:scatt-sent}.

    It remains only to check that the resulting sentences have the parameters
    stated in the corollary. Since the scatter sentence has distance rank
    \((k,q)\), by definition \(t\le k+q\). The same definition also gives an
    auxiliary \(i\in\{1,\ldots,q\}\) such that \(\beta(x)\) is local of distance
    rank \((k+i,q-i)\) and
    \[
        4\rho^-(k+i,q-i)\le r\le
        9^{k+i}\rho^-(k+i,q-i)\le \rho^-(k,q).
    \]
    This inequality gives the required bound \(r\le \rho^-(k,q)\).
    Repeatedly applying \Cref{obs:preservedistrank} shows that the same formula \(\beta(x)\) has distance
    rank \((k+1,q-1)\). Finally, \Cref{lem:ltp-local-subgraph} says that
    \(\beta(x)\) is semantically \(\rho^-(k+i,q-i)\)-local, hence semantically
    \(r/4\)-local by the left inequality above.
\end{proof}


\section{Proof of the locality theorem}\label{sec:proof}

\subsection{Separation lemma}

\begin{lemma}\label{lem:far-ltp}
    Fix $k,q\in\N$ and let $\phi(\tup x,\tup y)$ be a local \distFO formula of distance rank $(k,q)$.
    The formula
    $$\Big(\dist(\tup x,\tup y)>\rho^-(k,q)\Big) \ \land\ \phi(\tup x,\tup y)$$
    is equivalent to a disjunction of formulas of the form
    $$\Big(\dist(\tup x,\tup y)>\rho^-(k,q)\Big) \ \land\  \alpha(\tup x)\land \beta(\tup y),$$
    where $\alpha$ and $\beta$ are local \distFO formulas of distance rank $(k,q)$.
    Moreover, this disjunction can be effectively computed, given $k,q$, and $\phi$.
\end{lemma}

\begin{proof}
 We prove that assuming the background axiom
 \begin{align}\label{eq:far-ass}
\dist(\tup x,\tup y)>\rho^-(k,q),
 \end{align}
 the formula
 $\phi(\tup x;\tup y)$
is equivalent to a disjunction of
formulas of the form $\alpha(\tup x)\land \beta(\tup y)$, where
$\alpha$ and $\beta$ are local \distFO formulas of distance rank $(k,q)$.
The only nontrivial properties of the functions \(\rho^-\), \(\rho^+\) we will use are
\(\rho^+(k,q)\ge \rho^-(k,q)\ge 1\), which follows from the second condition in
\eqref{eq:main-property-rho}, and, for $q\ge 1$,
\begin{equation}\label{eq:prop-of-rho}
    \rho^-(k,q)-\rho^-(k+1,q-1) \ge \rho^+(k+1,q-1),
\end{equation}
which is precisely the first condition in \eqref{eq:main-property-rho}.

We proceed by structural induction on~$\phi$.

In the base case, $\phi(\bar x,\bar y)$
is an atomic formula, a binary distance atom with radius up to $\rho^-(k,q)$, or a unary distance atom with radius up to $\rho^-(k,q)$.
Unary distance atoms involve only one variable, hence they belong entirely to the $\tup x$-side or to the $\tup y$-side.
The same is true for first-order atoms and binary distance atoms which involve only variables from $\tup x$, or only variables from $\tup y$.
Otherwise, $\phi(\bar x,\bar y)$ evaluates to false, by \eqref{eq:far-ass} and $\rho^-(k,q)\ge 1$.

\newcommand{\far}{\mathrm{far}}
\newcommand{\near}{\mathrm{near}}
Consider the induction step. If $\phi$ is a a disjunction of two simpler formulas, the statement follows trivially by induction.
If $\phi$ is the negation of a simpler formula, we can use
$\neg\bigvee_{i}(\alpha_i\land \beta_i) \iff  \bigwedge_i(\neg\alpha_i\lor\neg\beta_i)$ and then distribute. Conjunction can be reduced to negation and disjunction.

Finally, suppose that
$$\phi(\tup x;\tup y)=\exists z~\Big(\dist(z,\tup x\tup y)\le \rho^+(k+1,q-1)\Big)\land \psi(\tup x \tup y z),$$
where $\psi(\tup x,\tup y z)$
has distance rank $(k+1,q-1)$ and $q\ge 1$.
Then
$\phi(\tup x;\tup y)$ is equivalent  to the disjunction of  the following formulas:
\begin{enumerate}
    \item
 $\exists z~ (\dist(z,\tup x)\le \rho^-(k+1,q-1))\land   \psi(\tup x \tup y z)$,
 \item
 $\exists z~(\dist(z,\tup y)\le \rho^-(k+1,q-1))\land \psi(\tup x \tup y z)$,
 \item
 $\exists z~\Big(\rho^-(k+1,q-1)< \dist(z,\tup x \tup y)\le\rho^+(k+1,q-1)\Big)\land \psi(\tup x \tup y z)$.
\end{enumerate}

In case 1, the condition
$\dist(z,\tup x)\le \rho^-(k+1,q-1)$
together with \eqref{eq:far-ass}, \eqref{eq:prop-of-rho}
and the triangle inequality imply
$$\dist(z,\tup y)\ge \dist(\tup x,\tup y)-\dist(z,\tup x)> \rho^-(k,q)-\rho^-(k+1,q-1) \ge \rho^+(k+1,q-1) \ge \rho^-(k+1,q-1).$$
 We can therefore use the inductive assumption to rewrite $\psi(\tup x \tup y z)$
 into a disjunction of formulas $\alpha(\tup x z)\land \beta(\tup y)$,
 where $\alpha$ and $\beta$ are local \distFO formulas of distance rank $(k+1,q-1)$.
 Therefore, using the fact that existential quantifications and disjunctions commute,
the formula in case 1
is equivalent to a disjunction of formulas of the form
$$\exists z~
\Bigl(\dist(z,\tup x)\le \rho^-(k+1,q-1)\Bigr)\land \alpha(\tup xz)\land \beta(\tup y),$$
 which is equivalent to the formula
 $$\beta(\tup y)\ \land\  \exists z~ \Bigl(\dist(z,\tup x)\le \rho^-(k+1,q-1)\Bigr)\land \alpha(\tup x z).$$
 Each disjunct is (equivalent to) a local \distFO formula of distance rank $(k,q)$, finishing case 1.
 Case 2 follows analogously.

 In case 3, we apply the inductive assumption to $\psi$
 and rewrite it into an equivalent disjunction of formulas of the form
$$\exists z~\Big(\rho^-(k+1,q-1)<\dist(z,\tup y \tup x)\le \rho^+(k+1,q-1)\Big)\land \alpha(\tup x\tup y)\land \beta(z),$$
where $\alpha$ and $\beta$ are local \distFO formulas of distance rank $(k+1,q-1)$.
This again uses the fact that existential quantifications and disjunctions commute.
The above formula is equivalent to
\begin{align}\label{eq:w}
\alpha(\tup x\tup y)\land
\exists z~\Big(\rho^-(k+1,q-1)<\dist(z,\tup y \tup x)\le \rho^+(k+1,q-1)\Big)\land \beta(z).
\end{align}
We argue that, assuming \eqref{eq:far-ass}, the following formulas are equivalent:

\begin{subequations}\label{eq:eq1}
\begin{align}
  &\rho^-(k+1,q-1)<\dist(z,\tup x\tup y)\le \rho^+(k+1,q-1), \label{eq:eq1a}\\[10pt]
  &\rho^-(k+1,q-1)<\dist(z,\tup x)\le \rho^+(k+1,q-1)  \notag\\
  \lor\quad &\rho^-(k+1,q-1)<\dist(z,\tup y)\le \rho^+(k+1,q-1).
   \label{eq:eq1b}
\end{align}
\end{subequations}
Indeed, clearly, \eqref{eq:eq1a} implies \eqref{eq:eq1b}. For the converse,
    suppose (by symmetry) that
    $\rho^-(k+1,q-1)<\dist(z,\tup x)\le \rho^+(k+1,q-1)$ holds.
    Using the triangle inequality and \eqref{eq:prop-of-rho}, we obtain
    $$\dist(z,\tup y)\ge \dist(\tup x,\tup y)-\dist(z,\tup x)>\rho^-(k,q)-\rho^+(k+1,q-1) \ge \rho^-(k+1,q-1),$$
    which implies \eqref{eq:eq1a}.

Therefore, the formula \eqref{eq:w}
is equivalent to
\begin{align}\label{eq:almostdone}
\alpha(\tup x\tup y)\land \exists z~& \Big(\rho^-(k+1,q-1)<\dist(z,\tup x )\le \rho^+(k+1,q-1)\Big)\land  \beta(z) \notag \\\lor\quad
\alpha(\tup x\tup y)\land \exists z~& \Big(\rho^-(k+1,q-1)<\dist(z,\tup y)\le \rho^+(k+1,q-1)\Big)\land  \beta(z).
\end{align}

We may apply the inductive assumption to the formula $\alpha(\tup x,\tup y)$ of distance rank \((k+1,q-1)\) to obtain an equivalent
disjunction of formulas of the form \(\alpha'(\bar x) \land \beta'(\bar y)\).
By basic distributivity laws, \eqref{eq:almostdone} is equivalent to a disjunction of formulas of the forms
\begin{align*}
& \alpha'(\tup x) \land \beta'(\tup y)\land \exists z~ \Big(\rho^-(k+1,q-1)<\dist(z,\tup x )\le \rho^+(k+1,q-1)\Big)\land  \beta(z), \\
& \alpha'(\tup x) \land \beta'(\tup y)\land \exists z~ \Big(\rho^-(k+1,q-1)<\dist(z,\tup y )\le \rho^+(k+1,q-1)\Big)\land  \beta(z).
\end{align*}

This construction is effective and yields the conclusion of the lemma.
 \end{proof}

\subsection{Far quantification lemma}
The following key lemma establishes an important special case of \Cref{thm:localglobal}.
\begin{lemma}\label{lem:submain}
    Fix $k,q\ge 1$, and let $\beta(x)$ be a local \distFO formula of  distance rank $(k+1,q-1)$ and let $\bar y$ be a set of $k$ variables.
    Then the formula
    $$\exists x~\Big(\dist(x,\tup y)>\rho^-(k+1,q-1)\Big)\ \land\  \beta(x)$$
    is equivalent to a boolean combination of local \distFO formulas of
     distance rank $(k,q)$ and of scatter sentences of distance rank $(k,q)$.
    Moreover, this boolean combination can be effectively computed, given $k,q$, and $\beta$.
\end{lemma}

The proof of \Cref{lem:submain}
relies on a combinatorial argument which is encapsulated in \Cref{lem:scatter} below.
We start with the following observation, also made by \cite[Lem. 3]{learningLogic},
which is similar in spirit to the Vitali covering lemma.

\begin{lemma}[{\cite{learningLogic}}]\label{lem:vitali}
    Given $r,c\in\N$, a graph, and a set $K$ of its vertices, there are $R\in\N$ and a set \(K' \subseteq K\),
     such~that
    \begin{enumerate}
        \item $r\le R\le r\cdot (c+1)^{|K|-1}$,
        \item the vertices in \(K'\) have pairwise distance larger than \(c R\), and
        \item every vertex in \(K\) has distance at most \(R\) to \(K'\).
    \end{enumerate}
\end{lemma}
\begin{proof}
For  \(t=0,1,2,\ldots\), we construct a subset \( K_t \subseteq K \) of size \(|K|-t\)
such that every vertex in \(K\) has distance at most \(R(t)\eqdef  r (c + 1)^{t} \) to \(K_t\).
Set \(K_0 \eqdef  K\).
For $t\ge 0$, if every pair of vertices in \(K_t\) has distance larger than \(c R(t)\), we are done.
Otherwise, pick vertices \(u,v \in K_t\) at distance at most \(c R(t)\).
Set \(K_{t+1} \eqdef  K_t - \{u\}\),
and observe that every vertex that has distance \(R(t)\) to \(u\) has distance \(R(t) + c R(t) = R(t+1)\) to \(v\).
This process terminates at the latest when \(t=|K|-1\).
\end{proof}

\begin{corollary}\label{cor:clusters}
    Fix $r,k\in\N$.
    Let \(G\) be a graph and \(\bar a\in V(G)^k\). There is some $R$ with \(r \le R \le 9^{k-1} r\)
    and a partition of $N_r^G(\tup a)$ into \emph{clusters} with the following properties:
    \begin{itemize}
        \item every cluster is a union of sets of the form $N_r^G(a_i)$, for $i\in [k]$,
        \item every cluster is contained in a ball of radius \(2R\), and
        \item clusters have pairwise distance larger than \(4R\).
    \end{itemize}
    Moreover the number $R$
    and the partition of  $\set{a_1,\ldots, a_k}$ into clusters
    can be determined by tests of the form $\dist^G(a_i,a_j)\le d$, for $i,j\in[k]$ and $d\le 8\cdot 9^{k-1} r$.
\end{corollary}
\begin{claimproof}
        Apply
          \Cref{lem:vitali} to \(K\eqdef\set{a_1,\ldots,a_k}\) and $c\eqdef 8$.
          We get \(I \subseteq [k]\), \(m : [k] \to I\), and  \(r \le R \le 9^{k-1} r\) such that
        \(\dist^G(a_i,a_j) > 8 R\) for all \(i\neq j \in I\), and
        \(\dist^G(a_i,a_{m(i)}) \le R\) for all \(i \in [k]\).
        As the clusters, take the sets \(\bigcup_{j \in m^{-1}(i)} N^G_r(a_j) \), for $i\in I$.

        Then every cluster is contained in a ball of radius \(R+r \le 2R\),
and clusters have pairwise distance larger than \(8R-2r-2R \ge 4R\).

For the ``moreover'' part, observe that for every fixed set $I\subset [k]$,mapping $m:[k]\to I$, and radius $R$ with $r\le R\le 9^{k-1} r$,
the conditions
\(\dist^G(a_i,a_j) > 8 R\)
        and \(\dist^G(a_i,a_{m(i)}) \le R\) can be verified
using tests of the form $\dist^G(a_i,a_j)\le d$, for $i,j\in[k]$ and $d\le 9^{k-1}\cdot 8\cdot r$.
\end{claimproof}



The following lemma is
the combinatorial counterpart of \Cref{lem:submain},
and forms
the core of the
proof of \Cref{thm:localglobal}.

\begin{lemma}\label{lem:scatter}
    Let \(r,k \in \N\) with $k\ge 1$.
    Let \(G\) be a graph, \(\bar a\in V(G)^k\) a tuple of its  vertices, and
    let $R$ with \(r \le R \le 9^{k-1} r\) and the partition of $N_r^G(\bar a)$ into clusters be as in \Cref{cor:clusters}.
    Let $X\subset V(G)$, and let $s$ be the size of some inclusion-wise maximal $4R$-scattered subset of $X$. Let $H\ge 4R+r$.
    Then the following conditions are equivalent:
\begin{enumerate}
    \item there is some $x\in X$ with $\dist(x,\tup a)>r$,

    \item
    \begin{enumerate}
        \item either there is some $x\in X$ with $r<\dist(x,\tup a)\le H$,

        \item or less than $s$ clusters intersect $X$.
    \end{enumerate}
\end{enumerate}
\end{lemma}
\begin{proof}
    Suppose first that
there is some $x\in X$ with $r<\dist(x,\bar a)\le 4R+r$.
In this case, clearly both conditions (1) and (2) are satisfied,
as $4R+r\le H$.
Assume therefore that this is not the case, that is,
 \begin{align}\label{eq:margin}
\dist(x,\tup a)\le 4R+r \quad\Longrightarrow\quad \dist(x,\tup a)\le  r\qquad\text{for all $x\in X$.}
 \end{align}
We show that then conditions (1) and (2b) are equivalent.
Since clearly (2a) implies (1), in fact (1) and (2) are equivalent in this case, proving the lemma.

Let $S$ be an inclusion-wise maximal $4R$-scattered subset of $X$ of size $s$.
In particular, $X\subseteq N_{4R}(S)$.

We show that each cluster $C$ which intersects $X$ contains exactly one element of $S$.
Indeed, suppose $v\in C\cap X$. As $v\in X\subseteq N_{4R}(S)$,
there is some $s\in S$ with $\dist(s,v)\le 4R$ and so $$\dist(s,C)\le 4R.$$
Since $v\in C\subset N_r(\bar a)$,
this
  implies $\dist(s,\bar a)\le 4R+r$.
  By \eqref{eq:margin}, we have that $\dist(s,\bar a)\le r$, in other words, $s\in N_r(\bar a)$.
  As $N_r(\bar a)$ is the union of all clusters,
   $s$ belongs to some cluster.
  Since the clusters have pairwise distance larger than $4R$ and $\dist(s,C)\le 4R$,
  we conclude that $s\in C$.
  Moreover, because $C$ is contained in a ball of radius $2R$, it cannot contain two elements of $S$, so $|S\cap C|=1$.

  By the above,
the number of
  clusters intersected by $X$ is equal to  $|S\cap N_r(\bar a)|$.
Condition (2b) is thus equivalent to $|S\cap N_r(\bar a)|<|S|$. We show that this is equivalent to condition (1).

Suppose
condition (1) holds. By \eqref{eq:margin},
there is
some $x\in X$ with $\dist(x,\bar a)>4R+r$.
Therefore,
$(S\cap N_r(\bar a))\cup \set{x}$
is a $4R$-scattered subset of $X$,
so $S\cap N_r(\bar a)$ is not an inclusion-wise maximal $4R$-scattered subset of $X$,
which implies $|S\cap N_r(\bar a)|<|S|$.

Conversely, if $|S\cap N_r(\bar a)|<|S|$ then
 there is some
$s\in S-  N_r(\bar a)$.
In particular, condition (1) holds.
Thus, assuming \eqref{eq:margin},
conditions (1) and (2b) are equivalent.
\end{proof}

\begin{proof}[Proof of \Cref{lem:submain}]
    Set $r\eqdef \rho^-(k+1,q-1)$.
    Let $A$ be a structure and $\tup a\in V(A)^k$.
    By applying \Cref{cor:clusters} to the Gaifman graph of $A$ and the tuple $\tup a$, we obtain
    a number $R$ with $r\le R\le 9^{k-1}r$ and partition $\cal P$ of $\set{a_1,\ldots,a_k}$ into clusters as in the statement of the lemma. Moreover, some
    $R$ and $\cal P$ with the required properties
    can be determined
    by performing tests of the form
    $\dist(a_i,a_j)\le d$ for $d<9^{k-1}\cdot 8\cdot r$. We have
    $$9^{k-1} \cdot 8\cdot  r\le 9^{k}\cdot  \rho^-(k+1,q-1)\le \rho^+(k+1,q-1)\le \rho^-(k,q),$$
         where the middle inequality follows from the second condition in \eqref{eq:main-property-rho}
         and the last inequality follows from the first condition in \eqref{eq:main-property-rho}.
So such tests can be performed by binary distance atoms in a local \distFO formula of distance rank $(k,q)$.

Set $H\coloneqq \rho^+(k+1,q-1)$. Since $R\le 9^{k-1}r$ and
\(r=\rho^-(k+1,q-1)\),
\[
    4R+r\le (4\cdot 9^{k-1}+1)r\le 9^{k+1}r\le H,
\]
where the last inequality follows from the second condition in \eqref{eq:main-property-rho}.
    We argue that conditions (2a) and (2b) in \Cref{lem:scatter} can be effectively expressed
    as a boolean combination of local \distFO formulas and scatter sentences of distance rank \((k,q)\).
    By \Cref{lem:scatter},
    this is enough to
    prove \Cref{lem:submain}.
    Condition (2a)
    is expressed by
     a single local \distFO formula of distance rank $(k,q)$:
    \[
\exists x\,
\bigl(\dist(x,\bar y)\le \rho^+(k+1,q-1)\bigr)
\land
\bigl(\dist(x,\bar y)>r\bigr)
\land
\beta(x),
\]

    We hence focus on condition (2b).
    Let \(s=s_{A,4R,\beta}\) be the size of the inclusion-wise maximal \(4R\)-scattered subset of $X\eqdef\setof{a\in A}{A\models \beta(a)}$ as discussed in \Cref{sec:typedefs}.
    Since \(r\le R\le 9^{k-1}r\) and \(r=\rho^-(k+1,q-1)\),
\[
4\rho^-(k+1,q-1)=4r\le 4R
\le 4\cdot 9^{k-1}r
\le 9^{k+1}r
=
9^{k+1}\rho^-(k+1,q-1)
\le \rho^-(k,q).
\]
This verifies the scatter-radius condition with radius \(4R\) and \(i=1\), where the last inequality follows from \eqref{eq:main-property-rho}.
     Thus,
    scatter sentences of distance rank \((k,q)\) determine either the exact value of \(s\) or the information that \(s \ge k+1\).

    For each $P\in\cal P$,
    let $C_P\subset V(A)$ denote the cluster
    containing $P$.
    Then $C_P=\bigcup_{a\in P}{N_r^A(a)}$, so
    the cluster $C_P$ intersects $X$
    if and only if
    $$A\models \bigvee_{a_i\in P}\exists x~(\dist(x,a_i)\le r)\land \beta(x).$$
    Thus we can build local \distFO formulas of distance rank $(k,q)$ that count how many of the up to \(k\) clusters intersect \(X\).
    To check if this quantity is smaller than \(s\), as required by condition (2b), it is sufficient to know either the exact value of \(s\) or the information that \(s \ge k+1\).
    As discussed above, this information is encoded by scatter sentences of distance rank \((k,q)\).
    %
    %
    %
\end{proof}

\subsection{Proof of the locality theorem}

We now prove \Cref{thm:localglobal}, repeated below.
\localglobalgame*

\begin{proof}
    We prove the statement by structural induction on $\phi(\tup x)$.
    Atoms of first-order logic, binary distance atoms, and unary distance atoms are local \distFO formulas, and boolean combinations are immediate.
    The case of local existential quantification is also straightforward: apply the inductive assumption to the quantified subformula, pull scatter sentences outside the quantifier, and distribute; the guarded existential quantification of each resulting local formula is again local.
    The only nontrivial case is when
    $\phi(\bar x)$ is of the form
    \begin{align}\label{eq:form1}
    \exists y~ \psi(y,\bar x),
    \end{align}
    where $\psi$ is a \distFO formula of distance rank $(k+1,q-1)$ which,
    by inductive assumption,
    may be assumed to be a boolean combination of scatter sentences and
    local \distFO formulas of distance rank $(k+1,q-1)$.
    Since scatter sentences can be pulled outside the existential quantifier
    (and \Cref{obs:preservescatterrank} allows to assume that they have distance rank $(k,q)$),
    and local \distFO formulas are closed under boolean combinations, we may furthermore assume  that
    $\psi(y,\bar x)$
    is a single local \distFO formula of distance rank $(k+1,q-1)$.

    In the case when $\bar x$ is empty, $\psi(y)$ is a local \distFO formula of distance rank $(k+1,q-1)$, and $\phi$ is a sentence which is equivalent to the scatter sentence $\textbf{scatter}(4\rho^-(k+1,q-1),\psi(y),1)$ of distance rank $(k,q)$, and we are done.
    We hence assume that $\bar x$ is non-empty, so $k\ge 1$.

Clearly, as $\rho^+(k+1,q-1)\ge \rho^-(k+1,q-1)$, the formula \eqref{eq:form1}
is equivalent to the disjunction
\begin{align}
    \exists y&\,\Big(\dist(y,\bar x)\le\rho^+(k+1,q-1)\Big)\land \psi(y,\bar x)
    \label{eq:form2}\\\lor\notag&\\   \label{eq:form3}
    \exists y&\,\Big(\dist(y,\bar x)>\rho^-(k+1,q-1)\Big)\land \psi(y,\bar x).
    \end{align}
The formula \eqref{eq:form2}
already is a local \distFO formula
of distance rank $(k,q)$.
By \Cref{lem:far-ltp},
the formula
\eqref{eq:form3}
is equivalent
to
a disjunction of formulas of the form
$$\exists y~\Big(\dist(y,\bar x)>\rho^-(k+1,q-1)\Big)\land \alpha(\bar x)\land \beta(y),$$
where $\alpha$ and $\beta$ are local \distFO formulas of distance rank $(k+1,q-1)$.
In turn, such a formula is equivalent to
    $$\alpha(\bar x)\land \exists y~\Big(\dist(y,\bar x)>\rho^-(k+1,q-1)\Big)\land \beta(y).$$
The conclusion follows by \Cref{lem:submain}.
\end{proof}

\printbibliography

\end{document}